\documentclass[reqno]{amsart}
\usepackage{amsmath}
\usepackage{amssymb}
\usepackage{amsthm}
\usepackage{amsfonts}
\usepackage{mathtools}
\usepackage{bm}
\usepackage{bbold}
\usepackage{ulem}
\usepackage{fancyhdr}
\usepackage{hyperref}
\usepackage{framed}
\usepackage[pdftex,final]{graphicx}
\usepackage{tikz}
\usepackage{float}
\usepackage{verbatim} 
\usetikzlibrary{calc,arrows}

\usepackage[iso-8859-7]{inputenc}

\theoremstyle{plain}
\newtheorem{thm}{Theorem}[section]
\newtheorem{corollary}[thm]{Corollary}

\newtheorem{prop}[thm]{Proposition}
\newtheorem*{theorem*}{Theorem}
\theoremstyle{definition}
\newtheorem{rem}[thm]{Remark}

\numberwithin{equation}{section}

\newcommand{\mytilde}{\raise.17ex\hbox{$\scriptstyle\mathtt{\sim}$}}

\newcommand{\RgeO}{\ensuremath{\mathbb{R}_{ \geq 0}}}
\newcommand{\Rat}[1]{\ensuremath{\mathbb{R}^{#1}}}
\newcommand{\Ni}{\ensuremath{|\mathcal{N}_i|}}

\begin{document}

\author{D. Boskos}
\address{Department of Automatic Control, School of Electrical Engineering, KTH Royal Institute of Technology, Osquldas v\"ag 10, 10044, Stockholm, Sweden}
\email{boskos@kth.se}

\author{D. V. Dimarogonas}
\address{Department of Automatic Control, School of Electrical Engineering, KTH Royal Institute of Technology, Osquldas v\"ag 10, 10044, Stockholm, Sweden}
\email{dimos@kth.se}

\begin{abstract}
In this report we provide a decentralized robust control approach, which guarantees that connectivity of a multi-agent network is maintained when certain bounded input terms are added to the control strategy. Our main motivation for this framework is to determine abstractions for multi-agent systems under coupled constraints which are further exploited for high level plan generation.
\end{abstract}

\title{Robust Connectivity Analysis for Multi-Agent Systems}
\maketitle

\section{Introduction}
Cooperative control of multi-agent systems constitutes a highly active area of research during the last two decades. Typical objectives are the consensus problem, which is concerned with finding a protocol that achieves convergence to a common value \cite{RwBeAe05}, reference tracking \cite{AaMaAa10} and formation control \cite{JmEm07}. A common feature in the approach to these problems is the design of decentralized control laws in order to achieve a global goal.

In the case of mobile robot networks with limited sensing and communication ranges, connectivity maintenance plays a fundamental role  \cite{ZmEmPg11}. In particular, it is required to constrain the control input in such a way that the network topology remains connected during the evolution of the system. For instance, in \cite{JmEm07} the rendezvous and formation control problems are studied while preserving connectivity, whereas in \cite{DdKk08} swarm aggregation is achieved by means of a control scheme that guarantees both connectivity and collision avoidance.

In our approach we provide a control law for each agent comprising of a decentralized feedback component and a free input term, which ensures connectivity maintenance, for all possible free input signals up to a certain bound of magnitude. The motivation for this approach comes from distributed control and coordination of multi-agent systems with locally assigned Linear Temporal Logic (LTL) specifications. In particular, by virtue of the invariance and robust connectivity maintenance properties, it is possible to define well posed decentralized abstractions for the multi-agent system which can be exploited for motion planning. The latter problem has been studied in our recent work \cite{BdDd15} for the single integrator dynamics case.

In this work, we design a bounded control law which results in network connectivity of the system for all future times provided that the initial relative distances of interconnected agents and the free input terms satisfy appropriate bounds. Furthermore, in the case of a spherical domain, it is shown that adding an extra repulsive vector field near the boundary of the domain can also guarantee invariance of the solutions and simultaneously maintain the robust connectivity property. The latter framework enables the construction of finite abstractions for the single integrator case.

The rest of the report is organized as follows. Section 2 introduces basic notation and preliminaries. In Section 3, results on robust connectivity maintenance are provided and explicit controllers which establish this property are designed. In Section 4, the corresponding controllers are appropriately modified, in order to additionally guarantee invariance of the solution for the case of a spherical domain. We summarize the results and discuss possible extensions in Section 5.

\section{Preliminaries and Notation}

\subsection{Notation}

We use the notation $|x|$ for the Euclidean norm of a vector $x\in\Rat{n}$. For a matrix $A\in\Rat{m\times n}$ we use the notation $|A|:= \max\{|Ax|: x\in\Rat{n}\}$ for the induced Euclidean matrix norm and $A^{T}$ for its transpose. For two vectors $x,y\in\Rat{n}(=\Rat{n\times 1})$ we denote their inner product by $\langle x,y \rangle:=x^{T}y$. Given a subset $S$ of $\Rat{n}$, we denote by ${\rm cl}(S)$, ${\rm int}(S)$ and $\partial S$ its closure, interior and boundary, respectively, where $\partial S:={\rm cl}(S)\setminus{\rm int}(S)$. For $R>0$, we denote by $B(R)$ the closed ball with center $0\in\Rat{n}$ and radius $R$. Given a vector $x=(x^1,\ldots,x^n)\in\Rat{n}$ we define the component operators $c_l(x):=x^l$, $l=1,\ldots,n$. Likewise, for a vector $x=(x_1,\ldots,x_N)\in\Rat{Nn}$ we define the component operators $c_l(x):=(c_l(x_1),\ldots,c_l(x_N))\in\Rat{N}$, $l=1,\ldots,n$.

Consider a multi-agent system with $N$ agents. For each agent $i\in\{1,\ldots,N\}:=\mathcal{N}$ we use the notation $\mathcal{N}_i$ for the set of its neighbors and $\Ni$ for its cardinality. We also consider an ordering of the agent's neighbors which we denote by $j_1,\ldots,j_{\Ni}$. $\mathcal{E}$ stands for the undirected network's edge set and $\{i,j\}\in\mathcal{E}$ iff $j\in\mathcal{N}_i$. The network graph $\mathcal{G}:=(\mathcal{N},\mathcal{E})$ is connected if for each $i,j\in \mathcal{N}$ there exists a finite sequence $i_1,\ldots,i_l\in\mathcal{N}$ with $i_1=i$, $i_l=j$ and $\{i_k,i_{k+1}\}\in\mathcal{E}$, for all $k=1,\ldots,l-1$. Consider an arbitrary orientation of the network graph $\mathcal{G}$, which assigns to each edge $\{i,j\}\in\mathcal{E}$ precisely one of the ordered pairs $(i,j)$ or $(j,i)$. When selecting the pair $(i,j)$ we say that $i$ is the tail and $j$ is the head of edge $\{i,j\}$. By considering a numbering $l=1,\ldots,M$ of the graph's edge set we define the $N\times M$ incidence matrix $D(\mathcal{G})$ corresponding to the particular orientation as follows:
$$
D(\mathcal{G})_{kl}:=\left\lbrace\begin{array}{ll} 1, & \textup{if vertex}\;  k\; \textup{is the head of edge} \; l \\  -1, & \textup{if vertex} \;  k \; \textup{is the tail of edge}\; l \\  0, & \textup{otherwise} \end{array}\right.
$$

\noindent The graph Laplacian $L(\mathcal{G})$ is the $N\times N$ positive semidefinite symmetric matrix $L(\mathcal{G}):=D(\mathcal{G})\times D(\mathcal{G})^T$. If we denote by $\mathbb{1}$ the vector $(1,\ldots,1)\in\Rat{N}$, then $L(\mathcal{G})\mathbb{1}=D(\mathcal{G})^T\mathbb{1}=0$. Let $0=\lambda_{1}(\mathcal{G})\le\lambda_{2}(\mathcal{G})\le\cdots\le\lambda_{N}(\mathcal{G})$ be the ordered eigenvalues of $L(\mathcal{G})$. Then each corresponding set of eigenvectors is orthogonal and $\lambda_{2}(\mathcal{G})>0$ iff $\mathcal{G}$ is connected.

\subsection{Problem Statement}

\noindent We focus on single integrator multi-agent systems with dynamics
\begin{equation}\label{multiagent:single:integrator}
\dot{x}_i=u_{i}, x_{i}\in\Rat{n}, i=1,\ldots,N
\end{equation}

\noindent We aim at designing decentralized control laws of the form

\begin{equation}\label{contol:law:connectedness}
u_{i}:=k_{i}(x_{i},x_{j_{1}},\ldots,x_{j_{|\mathcal{N}_{i}|}})+v_{i}
\end{equation}

\noindent which ensure that appropriate apriori bounds on the initial relative distances of interconnected agents guarantee network connectivity for all future times, for all free inputs $v_{i}$ bounded by certain constant. In particular, we assume that the network graph is connected as long as the maximum distance between two interconnected agents does not exceed a given positive constant $R$. In addition, we make the following connectivity hypothesis for the initial states of the agents.

\noindent \textbf{(ICH)} We assume that the agents' communication graph is initially connected and that
\begin{equation}\label{initial:relative:positions}
\max\{|x_{i}(0)-x_{j}(0)|:\{i,j\}\in\mathcal{E}\}\le\tilde{R}\;\textup{for certain constant}\; \tilde{R}\in (0,R)
\end{equation}

\subsection{Potential Functions}

We proceed by defining certain mappings which we exploit in order to design the control law \eqref{contol:law:connectedness} and prove that network connectivity is maintained. Let $r:\RgeO\to\RgeO$ be a continuous function satisfying the following property.

\noindent \textbf{(P)} $r(\cdot)$ is increasing and $r(0)>0$.

\noindent Also, consider the integral
\begin{equation} \label{function:P}
P(\rho)=\int_{0}^{\rho}r(s)sds,\rho\in\RgeO
\end{equation}

\noindent For each pair $(i,j)\in\mathcal{N}\times\mathcal{N}$ with $\{i,j\}\in\mathcal{E}$ we define the potential function $V_{ij}:\Rat{Nn}\to\RgeO$ as
\begin{equation} \label{Vij:dfn}
V_{ij}(x)=P(|x_{i}-x_{j}|),\forall x=(x_1,\ldots,x_N)\in\Rat{Nn}
\end{equation}

\noindent Notice that $V_{ij}(\cdot)=V_{ji}(\cdot)$. Furthermore, it can be shown that $V_{ij}(\cdot)$ is continuously differentiable and that
\begin{equation} \label{Vij:partial:xi2}
\frac{\partial}{\partial x_{i}}V_{ij}(x)=r(|x_{i}-x_{j}|)(x_{i}-x_{j})^{T}
\end{equation}

\noindent where $\frac{\partial}{\partial x_{i}}$ stands for the derivative with respect to the $x_{i}$-coordinates.

\begin{rem}
Notice that we are only interested in the values of the mappings $r(\cdot)$ and $P(\cdot)$ in the interval $[0,R]$, which stands for the maximum distance that two interconnected agents may achieve before losing connectivity. Yet, defining them on the whole positive line provides us certain technical flexibilities for the analysis employed in the proof of connectivity maintenance.
\end{rem}

\section{Connectivity Analysis}\label{Lyapunov:analysis}

\noindent In the following proposition we provide a control law \eqref{contol:law:connectedness} and an upper bound on the magnitude of the input terms $v_i(\cdot)$ which guarantee connectivity of the multi-agent network.

\begin{prop} \label{connectivity:maintainance}
For the multi-agent system \eqref{multiagent:single:integrator}, assume that (ICH) is fulfilled and define the control law
\begin{equation}\label{feedback:CaseI}
u_{i}=-\sum_{j\in\mathcal{N}_{i}}r(|x_{i}-x_{j}|)(x_{i}-x_{j})+v_{i}
\end{equation}

\noindent for certain continuous $r(\cdot)$ satisfying Property (P). Also, consider a constant $\delta>0$ and define
\begin{equation} \label{constant:K}
K:=\frac{2\sqrt{N(N-1)}|D(\mathcal{G})^{T}|}{\lambda_{2}(\mathcal{G})^{2}}
\end{equation}

\noindent where $D(\mathcal{G})$ is the incidence matrix of the systems' graph and $\lambda_{2}(\mathcal{G})$ the second eigenvalue of the graph Laplacian. We assume that
the positive constant $\delta$, the maximum initial distance $\tilde{R}$ and the function $r(\cdot)$ satisfy the restrictions
\begin{equation}  \label{constant:delta:general:constraint}
\delta\le\frac{1}{K}r(0)^{2}\frac{s}{r(s)},s\ge\tilde{R}
\end{equation}

\noindent with $K$ as given in \eqref{constant:K} and
\begin{equation} \label{Pat:Rtilde:vs:R}
MP(\tilde{R})\le P(R)
\end{equation}

\noindent where $P(\cdot)$ is given in \eqref{function:P}, and $M=|\mathcal{E}|$ is the cardinality of the system's graph edge set. Then, the system remains connected for all positive times, provided that the input terms $v_i(\cdot)$, $i=1,\ldots,N$ satisfy
\begin{equation} \label{input:individual:bound}
|v_{i}(t)|\le\delta,\forall t\ge 0
\end{equation}
\end{prop}

\begin{proof}
For the proof we follow parts of the analysis in \cite{JmEm07} (see also \cite[Section 7.2]{MmEm10}). Consider the energy function
\begin{equation} \label{V:dfn}
V:=\frac{1}{2}\sum_{i=1}^{N}\sum_{j\in\mathcal{N}_{i}}V_{ij}
\end{equation}

\noindent where the mappings $V_{ij}$, $\{i,j\}\in\mathcal{E}$ are given in \eqref{Vij:dfn}. Then it follows from \eqref{Vij:partial:xi2} that
\begin{equation} \label{V:partial:xi}
\frac{\partial}{\partial x_{i}}V(x)=\sum_{j\in\mathcal{N}_{i}}r(|x_{i}-x_{j}|)(x_{i}-x_{j})^{T}
\end{equation}

\noindent Also, in accordance with \cite[Section 7.2]{MmEm10} we have for $l=1,\ldots,n$ that
\begin{equation} \label{componentwise:feedback:rhs}
c_{l}\left(\sum_{j\in\mathcal{N}_{i}}r(|x_{i}-x_{j}|)(x_{i}-x_{j})\right)=L_{w}(x)c_{l}(x)
\end{equation}

\noindent The weighted Laplacian matrix $L_{w}(x)$ is given as
\begin{equation}\label{weighted:Laplacian}
L_{w}(x)=D(\mathcal{G})W(x)D(\mathcal{G})^{T}
\end{equation}

\noindent where $D(\mathcal{G})$ is the incidence matrix of the communication graph (see Notation) and
\begin{equation}\label{function:W}
W(x):={\rm diag}\{w_{1}(x),\ldots,w_{M}(x)\}:={\rm diag}\{r(|x_{i}-x_{j}|),\{i,j\}\in\mathcal{E}\}
\end{equation}

\noindent (recall that $M=|\mathcal{E}|$). Then, by evaluating the time derivative of $V$ along the trajectories of \eqref{multiagent:single:integrator}-\eqref{feedback:CaseI} and taking into account \eqref{V:dfn}, \eqref{V:partial:xi} and \eqref{componentwise:feedback:rhs} we get
\begin{align}
\dot{V}&=-\sum_{l=1}^{n}c_{l}\left(\frac{\partial}{\partial x}V(x)\right)c_{l}(\dot{x}) \nonumber \\
&=-\sum_{l=1}^{n}c_{l}(x)^{T}L_{w}(x)(L_{w}(x)c_{l}(x)-c_{l}(v)) \nonumber \\
&=-\sum_{l=1}^{n}c_{l}(x)^{T}L_{w}(x)^{2}c_{l}(x)+\sum_{l=1}^{n}c_{l}(x)^{T}L_{w}(x)c_{l}(v) \nonumber \\
&\le-\sum_{l=1}^{n}c_{l}(x)^{T}L_{w}(x)^{2}c_{l}(x)+\left|\sum_{l=1}^{n}c_{l}(x)^{T}L_{w}(x)c_{l}(v)\right| \label{Lyapunov:function:derivative}
\end{align}

We want to provide appropriate bounds for the right hand side of \eqref{Lyapunov:function:derivative} which can guarantee that the sign of $\dot{V}$ is negative whenever the maximum distance between two agents exceeds the bound $\tilde{R}$ on the maximum initial distance as given in \eqref{initial:relative:positions}. First, we provide certain useful inequalities between the eigenvalues of the weighted Laplacian  $L_{w}(x)$ and the Laplacian matrix of the graph $L(\mathcal{G})$. Notice, that due to \eqref{function:W}, for each $i=1,\ldots,M$ we have $w_{i}(x)=r(|x_{k}-x_{\ell}|)$ for certain $\{k,\ell\}\in\mathcal{E}$ and hence, by virtue of Property (P), it holds
\begin{equation}\label{wis:property}
0<r(0)\le w_{i}(x)\le \max_{\{k,\ell\}\in\mathcal{E}}r(|x_{k}-x_{\ell}|)
\end{equation}

\noindent From \eqref{wis:property}, it follows that $L_{w}(x)$ has precisely the same properties with those provided for $L(\mathcal{G})$ in the Notation subsection. Furthermore, it holds
\begin{equation}\label{lambda2}
\lambda_{2}(x)\ge \lambda_{2}(\mathcal{G})r(0)
\end{equation}

\noindent where $0=\lambda_{1}(x)<\lambda_{2}(x)\le\cdots\le\lambda_{N}(x)$ and $0=\lambda_{1}(\mathcal{G})<\lambda_{2}(\mathcal{G})\le\cdots\le\lambda_{N}(\mathcal{G})$ are the eigenvalues of $L_{w}(x)$ and the Laplacian matrix of the graph $L(\mathcal{G})$, respectively. Indeed, in order to show \eqref{lambda2}, notice that
\begin{align*}
L_{w}(x)&=D(\mathcal{G}){\rm diag}\{w_{1}(x),\ldots,w_{M}(x)\}D(\mathcal{G})^{T} \\
&=D(\mathcal{G}){\rm diag}\{r(0),\ldots,r(0)\}D(\mathcal{G})^{T} \\
&+D(\mathcal{G}){\rm diag}\{w_{1}(x)-r(0),\ldots,w_{M}(x)-r(0)\}D(\mathcal{G})^{T}=r(0)L(\mathcal{G})+B
\end{align*}

\noindent where \eqref{wis:property} implies that $B:=D(\mathcal{G}){\rm diag}\{w_{1}(x)-r(0),\ldots,w_{M}(x)-r(0)\}D(\mathcal{G})^{T}$ is positive semidefinite. Hence, it holds $L_{w}(x)\succeq r(0)L(\mathcal{G})$, with $\succeq$ being the partial order on the set of symmetric $N\times N$ matrices and thus, we deduce from Corollary 7.7.4(c) in \cite[page 495]{HrJc13} that \eqref{lambda2} is fulfilled.

\noindent In the sequel we introduce some additional notation. Let $H$ be the subspace
\begin{equation*}
H:=\{x\in\Rat{Nn}:x_{1}=x_{2}=\cdots=x_{N}\}
\end{equation*}

\noindent For a vector $x\in\Rat{Nn}$ we denote by $\bar{x}$ its projection to the subspace $H$, and $x^{\perp}$ its orthogonal complement with respect to that subspace, namely $x^{\perp}:=x-\bar{x}$. By taking into account that for all $y\in H$ we have $D(\mathcal{G})^Tc_l(y)=0$ and hence, due to \eqref{weighted:Laplacian}, that $c_l(y)\in\ker(L_w(x))$, it follows that for every vector $x\in\Rat{Nn}$ with $x=\bar{x}+x^{\perp}$ it holds
\begin{equation}\label{A}
L_{w}(x)c_{l}(\bar{x})=0\Rightarrow L_{w}(x)c_{l}(x)=L_{w}(x)c_{l}(\bar{x}+x^{\perp})=L_{w}(x)c_{l}(x^{\perp})
\end{equation}

\noindent We also denote by $\Delta x\in\Rat{Mn}$ the stack column vector of the vectors $x_{i}-x_{j},\{i,j\}\in\mathcal{E}$ with the edges ordered as in the case of the incidence matrix. It is thus straightforward to check that for all $x\in\Rat{Nn}$
\begin{equation}\label{Deltax:incidence:mat}
D(\mathcal{G})^Tx=\Delta x
\end{equation}

\noindent and furthermore, due to \eqref{function:W} and \eqref{wis:property}, that
\begin{equation}\label{B1}
|W(x)|\le r(|\Delta x|_{\infty})
\end{equation}

\noindent where
\begin{equation}
|\Delta x|_{\infty}:=\max\{|\Delta x_{i}|,i=1,\ldots,M\}
\end{equation}

\noindent Before proceeding we state the following elementary facts, whose proofs can be found in the Appendix. In particular, for the vectors  $x=(x_1,\ldots,x_N),y=(y_1,\ldots,y_N)\in\Rat{Nn}$ the following properties hold.

\noindent \textbf{Fact I.}
\begin{equation} \label{fact1}
|L_{w}(x)c_{l}(x^{\perp})|\ge \lambda_{2}(x)|c_{l}(x^{\perp})|,\forall l=1,\ldots,n
\end{equation}

\noindent \textbf{Fact II.}
\begin{equation} \label{fact2}
\sum_{l=1}^n|c_l(x)||c_l(y)|\le |x||y|
\end{equation}

\noindent \textbf{Fact III.}\label{fact3}
\begin{equation}
|x^{\perp}|\ge\frac{1}{\sqrt{2(N-1)}}|\Delta x|
\end{equation}

\noindent \textbf{Fact IV.}
\begin{equation}\label{fact4}
\sqrt{2}|x^{\perp}|\ge |\Delta x|_{\infty}
\end{equation}

\noindent We are now in position to bound the derivative of the energy function $V$ and exploit the result in order to prove the desired connectivity maintenance property. We break the subsequent proof in two main steps.

\noindent \textbf{Step 1: Bound estimation for the rhs of \eqref{Lyapunov:function:derivative}.}  \\

\noindent \textbf{Bound for the first term in \eqref{Lyapunov:function:derivative}.} By taking into account \eqref{A}, it follows that
\begin{equation} \label{first:term:bound1}
\sum_{l=1}^{n}c_{l}(x)^{T}L_{w}(x)^{2}c_{l}(x)=\sum_{l=1}^{n}\left|L_{w}(x)c_{l}(x)\right|^{2}=\sum_{l=1}^{n}\left|L_{w}(x)c_{l}(x^{\perp})\right|^{2}
\end{equation}

\noindent and by exploiting Fact I and \eqref{lambda2}, we get
\begin{align}
\sum_{l=1}^{n}\left|L_{w}(x)c_{l}(x^{\perp})\right|^{2}&\ge\sum_{l=1}^{n}\lambda_{2}(x)^2|c_{l}(x^{\perp})|^{2}\ge\sum_{l=1}^{n}[\lambda_{2}(\mathcal{G})r(0)]^2|c_{l}(x^{\perp})|^{2} \nonumber \\
&=[\lambda_{2}(\mathcal{G})r(0)]^{2}\sum_{l=1}^{n}\left|c_{l}(x^{\perp})\right|^{2}=[\lambda_{2}(\mathcal{G})r(0)]^{2}|x^{\perp}|^{2} \label{first:term:bound2}
\end{align}

\noindent Thus, it follows from \eqref{first:term:bound1} and \eqref{first:term:bound2} that
\begin{equation} \label{first:term:bound:final}
\sum_{l=1}^{n}c_{l}(x)^{T}L_{w}(x)^{2}c_{l}(x)\ge [\lambda_{2}(\mathcal{G})r(0)]^{2}|x^{\perp}|^{2}
\end{equation}

\noindent \noindent \textbf{Bound for the second term in \eqref{Lyapunov:function:derivative}.} For this term, we have from \eqref{weighted:Laplacian} and \eqref{Deltax:incidence:mat} that
\begin{align} \label{second:term:bound1}
\left|\sum_{l=1}^{n}c_{l}(x)^{T}L_{w}(x)c_{l}(v)\right|&\le\sum_{l=1}^{n}|c_{l}(x)^{T}L_{w}(x)c_{l}(v)| \nonumber \\
&=\sum_{l=1}^{n}|c_{l}(x)^{T}D(\mathcal{G})W(x)D(\mathcal{G})^{T}c_{l}(v)| \nonumber \\
&=\sum_{l=1}^{n}|c_{l}(\Delta x)^{T}W(x)D(\mathcal{G})^{T}c_{l}(v)| \nonumber \\
&\le\sum_{l=1}^{n}|c_{l}(\Delta x)||W(x)||D(\mathcal{G})^{T}||c_{l}(v)|
\end{align}

\noindent By taking into account \eqref{B1}, we obtain
\begin{equation} \label{second:term:bound2}
\sum_{l=1}^{n}|c_{l}(\Delta x)||W(x)||D(\mathcal{G})^{T}||c_{l}(v)|\le\sum_{l=1}^{n}|c_{l}(\Delta x)|r(|\Delta x|_{\infty})|D(\mathcal{G})^{T}||c_{l}(v)|
\end{equation}

\noindent Also, by exploiting Fact II, we get that
\begin{align} \label{second:term:bound3}
\sum_{l=1}^{n}|c_{l}(\Delta x)|r(|\Delta x|_{\infty}) & |D(\mathcal{G})^{T}||c_{l}(v)| \nonumber \\
= & r(|\Delta x|_{\infty})|D(\mathcal{G})^{T}|\sum_{l=1}^{n}|c_{l}(\Delta x)||c_{l}(v)| \nonumber \\
\le & r(|\Delta x|_{\infty})|D(\mathcal{G})^{T}||\Delta x||c_{l}(v)|  \nonumber \\
\le & r(|\Delta x|_{\infty})|D(\mathcal{G})^{T}||\Delta x|\sqrt{N}|v|_{\infty}
\end{align}

\noindent where
$$
|v|_{\infty}:=\max\{|v_i|,i=1,\ldots,N\}
$$

\noindent Hence, it follows from \eqref{second:term:bound1}-\eqref{second:term:bound3} that
\begin{equation} \label{second:term:bound:final}
\left|\sum_{l=1}^{n}c_{l}(x)^{T}L_{w}(x)c_{l}(v)\right|\le \sqrt{N}|D(\mathcal{G})^{T}||\Delta x|r(|\Delta x|_{\infty})|v|_{\infty}
\end{equation}

\noindent Thus, we get from \eqref{Lyapunov:function:derivative}, \eqref{first:term:bound:final} and \eqref{second:term:bound:final} that
\begin{equation*}
\dot{V}\le -[\lambda_{2}(\mathcal{G})r(0)]^{2}|x^{\perp}|^{2}+\sqrt{N}|D(\mathcal{G})^{T}||\Delta x|r(|\Delta x|_{\infty})|v|_{\infty}
\end{equation*}

\noindent and by exploiting Facts III and IV, that
\begin{align*}
\dot{V}&\le - [\lambda_{2}(\mathcal{G})r(0)]^{2}\frac{1}{\sqrt{2(N-1)}}|\Delta x|\frac{1}{\sqrt{2}}|\Delta x|_{\infty}+\sqrt{N}|D(\mathcal{G})^{T}||\Delta x|r(|\Delta x|_{\infty})|v|_{\infty} \\
&=|\Delta x|\left(-\frac{1}{2\sqrt{N-1}}[\lambda_{2}(\mathcal{G})r(0)]^{2}|\Delta x|_{\infty}+\sqrt{N}|D(\mathcal{G})^{T}|r(|\Delta x|_{\infty})|v|_{\infty}\right)
\end{align*}

\noindent By using the notation $|\Delta x|_{\infty}:=s$, in order to guarantee that the above rhs is negative for $s\ge\tilde{R}$, it should hold
\begin{align*}
&-\frac{\lambda_{2}(\mathcal{G})^2}{2\sqrt{(N-1)}}r(0)^{2}s+\sqrt{N}|D(\mathcal{G})^{T}|r(s)|v|_{\infty}\le 0,\forall s\ge\tilde{R}\iff \\
&\frac{2\sqrt{N(N-1)}|D(\mathcal{G})^{T}|}{\lambda_{2}(\mathcal{G})^{2}}|v|_{\infty}\le r(0)^{2}\frac{s}{r(s)},\forall s\ge\tilde{R}
\end{align*}

\noindent or equivalently
\begin{equation} \label{v:infinity:norm:bound}
|v|_{\infty}\le\frac{1}{K}r(0)^{2}\frac{s}{r(s)},\forall s\ge\tilde{R}
\end{equation}

\noindent with $K$ as given in \eqref{constant:K}. Hence, we have shown that for $v$ satisfying \eqref{v:infinity:norm:bound} the following implication holds
\begin{equation} \label{Vderivative:negative:condition}
|\Delta x|_{\infty}\ge\tilde{R}\Rightarrow \dot{V}\le 0
\end{equation}

\noindent \textbf{Step 2: Proof of connectivity.}  \\

By assuming that conditions \eqref{input:individual:bound}, \eqref{constant:delta:general:constraint} and \eqref{Pat:Rtilde:vs:R} in the statement of the proposition are fulfilled and recalling that according to (ICH) \eqref{initial:relative:positions} holds, we can show that the system will remain connected for all future times. Indeed, let $x(\cdot)$ be the solution of the closed loop system \eqref{multiagent:single:integrator}-\eqref{feedback:CaseI} with initial condition satisfying \eqref{initial:relative:positions}, defined on the maximal right interval $[0,T_{\max})$. We claim that the system remains connected on $[0,T_{\max})$, namely, that $\max\{|x_i(t)-x_j(t)|:\{i,j\}\in\mathcal{E}\}\le R$ for all $t\in [0,T_{\max})$, which by boundedness of the dynamics on the set $\mathcal{F}:=\{x\in\Rat{Nn}:|x_i-x_j|\le R,\forall \{i,j\}\in\mathcal{E}\}$ implies that $T_{\max}=\infty$. In order to prove the last assertion, assume on the contrary that $T_{\max}<\infty$. Then, by taking into account that $x(t)$ remains in $\mathcal{F}$ for all $t\in[0,T_{\max})$ and that the dynamics are bounded on $\mathcal{F}$, it follows that $x(t)$ remains in a compact subset of $\Rat{Nn}$ for all $t\in[0,T_{\max})$ and hence, that it can be extended, contradicting maximality of $[0,T_{\max})$. We proceed with the proof of connectivity. First, notice that
\begin{equation} \label{Vat:zero}
V(x(0))\le\frac{1}{2}P(R)
\end{equation}

\noindent Indeed, by exploiting \eqref{initial:relative:positions} and \eqref{Pat:Rtilde:vs:R} we get that
\begin{align}
V(x(0)) & =\frac{1}{2}\sum_{i=1}^{N}\sum_{j\in\mathcal{N}_{i}}P(|x_{i}(0)-x_{j}(0)|) \nonumber \\
& \le\frac{1}{2}\sum_{i=1}^{N}\sum_{j\in\mathcal{N}_{i}}P(\tilde{R})=\frac{M}{2}P(\tilde{R})\le\frac{1}{2}P(R)  \label{Vat:zero:proof}
\end{align}

\noindent In order to prove our claim, it suffices to show that
\begin{equation} \label{Voft:bound}
V(x(t))\le\frac{1}{2}P(R),\forall t\in [0,T_{\max})
\end{equation}

\noindent  because if $|x_{i}(t)-x_{j}(t)|>R$ for certain $t\in [0,T_{\max})$ and $\{i,j\}\in\mathcal{E}$, then $V(x(t))\ge \frac{1}{2}P(|x_{i}(t)-x_{j}(t)|)>\frac{1}{2}P(R)$. We prove \eqref{Voft:bound} by contradiction. Indeed, suppose on the contrary that there exists $T\in (0,T_{\max})$ (due to \eqref{Vat:zero}) such that
\begin{equation} \label{VatT}
V(x(T))>\frac{1}{2}P(R)
\end{equation}

\noindent and define
\begin{equation} \label{time:tau}
\tau:=\min\{t\in[0,T]:V(x(\bar{t}))>\tfrac{1}{2}P(R),\forall\bar{t}\in(t,T]\}
\end{equation}

\noindent which due to \eqref{VatT} and continuity of $V(x(\cdot))$ is well defined. Then it follows from \eqref{Vat:zero} and \eqref{time:tau} that
\begin{equation} \label{Vat:tau:toT:properties}
V(x(\tau))=\frac{1}{2}P(R),V(x(t))>\frac{1}{2}P(R),\forall t\in(\tau,T]
\end{equation}

\noindent hence, there exists $\bar{\tau}\in(\tau,T)$ such that
\begin{equation} \label{Vderivative:at:taubar}
\dot{V}(x(\bar{\tau}))=\frac{V(x(T))-V(x(\tau))}{T-\tau}>0
\end{equation}

\noindent On the other hand, due to \eqref{Vat:tau:toT:properties}, it holds
\begin{equation} \label{Vattaubar}
V(x(\bar{\tau}))>\frac{1}{2}P(R)
\end{equation}

\noindent which implies that there exists $\{i,j\}\in\mathcal{E}$ with
\begin{equation} \label{xij:difference}
|x_{i}(\bar{\tau})-x_{j}(\bar{\tau})|>\tilde{R}
\end{equation}

\noindent Indeed, if \eqref{xij:difference} does not hold, then we can show as in \eqref{Vat:zero:proof} that $V(x(\bar{\tau}))\le\frac{1}{2}P(R)$ which contradicts \eqref{Vattaubar}. Notice that by virtue of \eqref{input:individual:bound} and \eqref{constant:delta:general:constraint}, \eqref{v:infinity:norm:bound} is fulfilled. Hence, we get from \eqref{xij:difference} that $|\Delta x(\bar{\tau})|_{\infty}>\tilde{R}$ and thus from \eqref{Vderivative:negative:condition} it follows that $\dot{V}(x(\bar{\tau}))\le 0$, which contradicts \eqref{Vderivative:at:taubar}. We conclude that \eqref{Voft:bound} holds and the proof is complete.
\end{proof}

In the following corollary, we apply the result of Proposition \ref{connectivity:maintainance} in order to provide two explicit feedback laws of the form \eqref{feedback:CaseI}, a linear and a nonlinear one and compare their performance in the subsequent remark.

\begin{corollary}
For the multi agent system \eqref{multiagent:single:integrator}, assume that (ICH) is fulfilled and consider the control law \eqref{contol:law:connectedness} as given by \eqref{feedback:CaseI}. By imposing the additional requirement $r(0)=r(\tilde{R})=1$ and defining
\begin{equation} \label{constant:delta}
\delta:=\frac{\tilde{R}}{K}
\end{equation}

\noindent with $\tilde{R}$ and $K$ as given in \eqref{initial:relative:positions} and \eqref{constant:K}, respectively, the system remains connected for all positive times, provided that the function $r(\cdot)$ and the constant $\tilde{R}$ are selected as in the following two cases (L) and (NL) (providing a linear and a nonlinear feedback, respectively).

\noindent \textbf{Case (L).} We select
\begin{equation} \label{function:rho:linear}
r(s):=1,s\ge 0
\end{equation}

\noindent and
\begin{equation} \label{constant:R:linear}
\tilde{R}\le\frac{1}{\sqrt{M}}R
\end{equation}

\noindent where $M$ is the cardinality of the system's graph edge set.

\noindent \textbf{Case (NL).} We select
\begin{equation} \label{function:rho:nonlinear}
r(s):=\left\lbrace\begin{array}{ll} 1, & s\in [0,\tilde{R}] \\ \frac{s}{\tilde{R}}, & s\in (\tilde{R},R] \\ \frac{R}{\tilde{R}}, & s\in (R,\infty) \end{array}\right.
\end{equation}

\noindent and
\begin{equation} \label{constant:R:nonlinear}
\tilde{R}\le\left(\frac{2}{3M-1}\right)^{\frac{1}{3}}R
\end{equation}
\end{corollary}

\begin{proof}
For the proof we apply the result of Proposition \eqref{connectivity:maintainance}. In particular, it suffices to show that for both cases (L) and (NL) the selection of the function $r(\cdot)$ and the initial maximum distance $\tilde{R}$ satisfy \eqref{constant:delta:general:constraint} and \eqref{Pat:Rtilde:vs:R}, with $\delta$ as given by \eqref{constant:delta}.

\noindent \textbf{Case (L).} Indeed, it follows from  \eqref{constant:delta} and \eqref{function:rho:linear} that
\begin{align*}
\delta=\frac{\tilde{R}}{K}=\frac{1}{K}\frac{\tilde{R}}{r(s)}r(0)^{2}\le\frac{1}{K}\frac{s}{r(s)}r(0)^{2},\forall s\ge\tilde{R}
\end{align*}

\noindent hence, \eqref{constant:delta:general:constraint} is fulfilled. Furthermore, it follows from \eqref{function:rho:linear} and \eqref{function:P} that
\begin{align*}
MP(\tilde{R})\le P(R) & \iff M\int_{0}^{\tilde{R}}r(s)sds\le\int_{0}^{R}r(s)sds \\
& \iff  M\int_{0}^{\tilde{R}}sds\le\int_{\tilde{R}}^{R}sds \iff M\frac{\tilde{R}^{2}}{2}\le\frac{R^{2}}{2}\iff \tilde{R}\le\frac{1}{\sqrt{M}}R
\end{align*}

\noindent \textbf{Case (NL).} Also in this case, it follows from \eqref{constant:delta} and \eqref{function:rho:nonlinear} that
\begin{align*}
\delta=\frac{\tilde{R}}{K}=\frac{\tilde{R}}{K}\frac{s}{s}r(0)^{2}\le\frac{1}{K}\frac{s}{\frac{s}{\tilde{R}}}r(0)^{2}=\frac{1}{K}\frac{s}{r(s)}r(0)^{2},\forall s\in (\tilde{R},R]
\end{align*}

\noindent and that

\begin{align*}
\delta&=\frac{\tilde{R}}{K}\le\frac{\tilde{R}}{K}\frac{s}{R}r(0)^{2} \\
&=\frac{1}{K}\frac{s}{\frac{s}{\tilde{R}}}r(0)^{2}\le\frac{1}{K}\frac{s}{\frac{R}{\tilde{R}}}r(0)^{2}=\frac{1}{K}\frac{s}{r(s)}r(0)^{2},\forall s\ge R
\end{align*}

hence, \eqref{constant:delta:general:constraint} is again fulfilled. In addition, it follows from \eqref{function:rho:nonlinear} and \eqref{function:P} that
\begin{align*}
MP(\tilde{R})\le P(R) & \iff M\int_{0}^{\tilde{R}}r(s)sds\le\int_{0}^{R}r(s)sds \\
& \iff  M\int_{0}^{\tilde{R}}r(s)sds\le\int_{0}^{\tilde{R}}r(s)sds+\int_{\tilde{R}}^{R}r(s)sds \\
& \iff (M-1)\int_{0}^{\tilde{R}}sds\le\frac{1}{\tilde{R}}\int_{\tilde{R}}^{R}s^{2}ds \iff (M-1)\frac{\tilde{R}^{2}}{2}\le \frac{1}{\tilde{R}}\left.\frac{1}{3}s^{3}\right|_{\tilde{R}}^{R} \\
& \iff \frac{M-1}{2}\tilde{R}^{3}\le\frac{1}{3}R^{3}-\frac{1}{3}\tilde{R}^{3} \iff \frac{3M-3}{6}\tilde{R}^{3}+\frac{2}{6}\tilde{R}^{3}\le \frac{1}{3}R^{3} \\
& \iff \frac{3M-1}{2}\tilde{R}^{3}\le R^{3}\iff\tilde{R}\le\left(\frac{2}{3M-1}\right)^{\frac{1}{3}}R
\end{align*}
\end{proof}

\begin{rem}
At this point we derive the advantage of using the nonlinear controller over the linear one by comparing the ratio of the maximal initial relative distance that maintains connectivity for these two cases. In both cases we have the same bound on the free input terms and the same feedback law up to some distance between neighboring agents, which allows us to compare their performance under the criterion of maximizing the largest initial distance between two interconnected agents. In particular, this ratio, which depends on the number of edges in the systems' graph, is given by
\begin{equation}
Rat(M):=\frac{\frac{1}{\sqrt{M}}}{\left(\frac{2}{3M-1}\right)^{\frac{1}{3}}}
\end{equation}

\noindent It is evident from the plot of $Rat(M)$ in Figure 1 that it is a decreasing function of $M$ with values less than 1 for $M\ge 1$. The latter property follows quite intuitively if we take a look at Figure 2. Indeed, as sown in the proof of Proposition \eqref{connectivity:maintainance}, both the maximal initial distance $\tilde{R}_L$ for the linear and $\tilde{R}_{NL}$ for the nonlinear case are expressed by virtue of \eqref{Pat:Rtilde:vs:R} as the solutions of the equation $\frac{P(\tilde{R})}{P(R)-P(\tilde{R})}=\frac{1}{M-1}$. In Figure 2, both the quotient of the area inside the orange frame over the area in the red frame and the quotient of the violet area over the blue area are the same and equal to $\frac{P(\tilde{R})}{P(R)-P(\tilde{R})}=\frac{1}{M-1}$, with $P(\cdot)$ as defined for both the cases (L) and (NL) through \eqref{function:P}, \eqref{function:rho:linear} and \eqref{function:rho:nonlinear}. It is thus straightforward that the quotient of $\tilde{R}_{L}$ over $\tilde{R}_{NL}$, namely $Rat(M)$ should be less than 1. Yet, we also provide a formal proof of this argument by studying the monotonicity of $Rat(\cdot)$.

\noindent For convenience we consider the 6-th power of $Rat(M)$, namely the function
\begin{equation*}
Rat^{6}(M):=\left[\frac{\frac{1}{\sqrt{M}}}{\left(\frac{2}{3M-1}\right)^{\frac{1}{3}}}\right]^{6}=\frac{(3M-1)^{2}}{{2^{3}M^{3}}}
\end{equation*}

\noindent Hence, by evaluating the derivative of $Rat^{6}(M)$ we obtain
\begin{align*}
\frac{d}{dM}Rat^{6}(M)& =\frac{d}{dM}\frac{(3M-1)^{2}}{{2^{3}M^{3}}}=\frac{2(3M-1)3\cdot2^{3}M^{3}-(3M-1)^{2}2^{3}3M^{2}}{2^{6}M{6}} \\
& = \frac{3\cdot2^{3}M^{2}(3M-1)[2M-(3M-1)]}{2^{6}M^{6}}=\frac{3}{2^{3}}\cdot\frac{3M-1}{M^{4}}(1-M)<0,\;{\rm for}\; M>1
\end{align*}
\end{rem}

\begin{figure}[h]
\begin{center}
\includegraphics[scale=0.45]{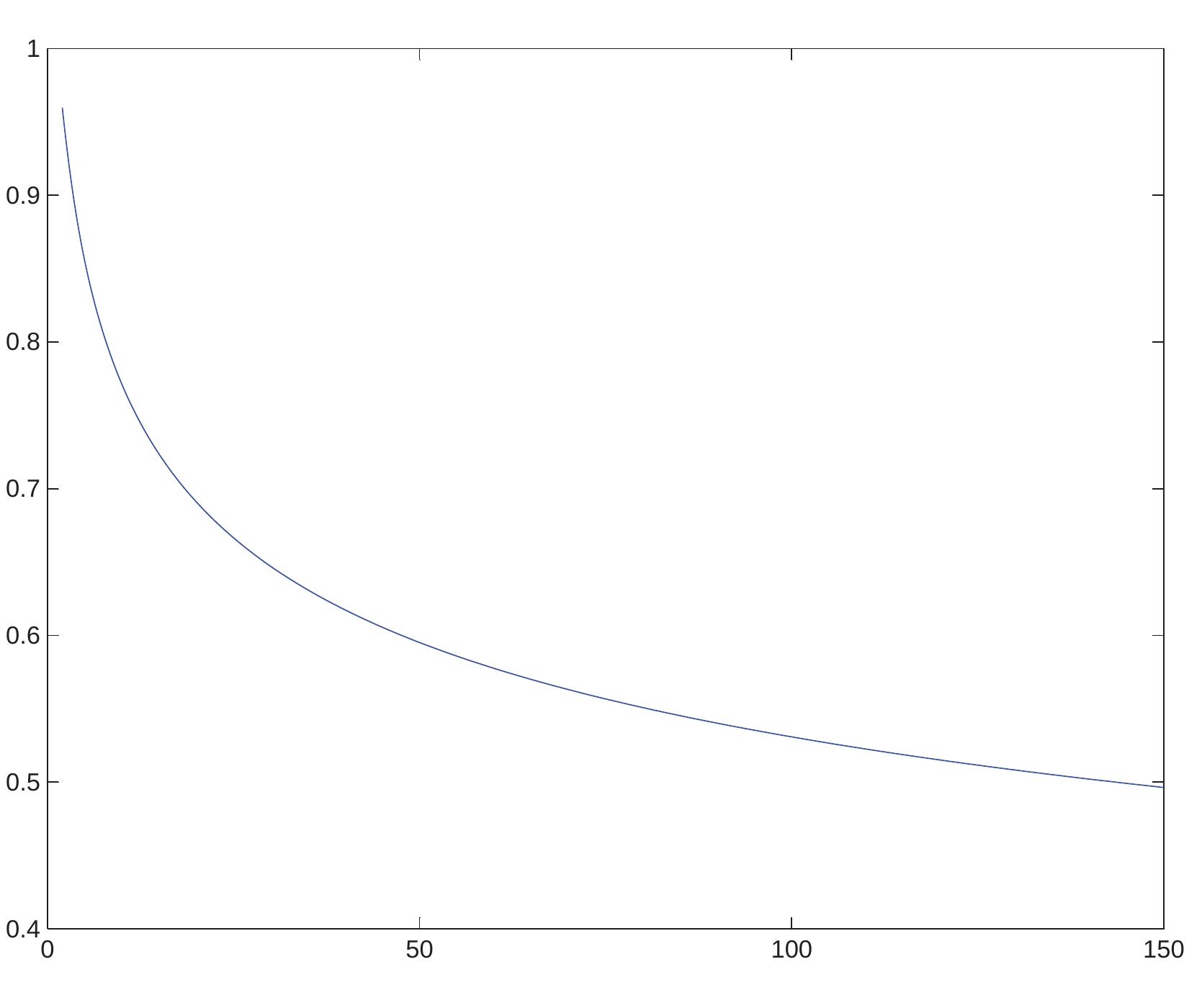}
\caption{\small \sl This figure shows the ratio $\frac{1}{\sqrt{M}}\;/\; \left(\frac{2}{3M-1}\right)^{\frac{1}{3}}$ for the number of edges ranging from 2 to 150 \label{figure1}}
\end{center}
\end{figure}

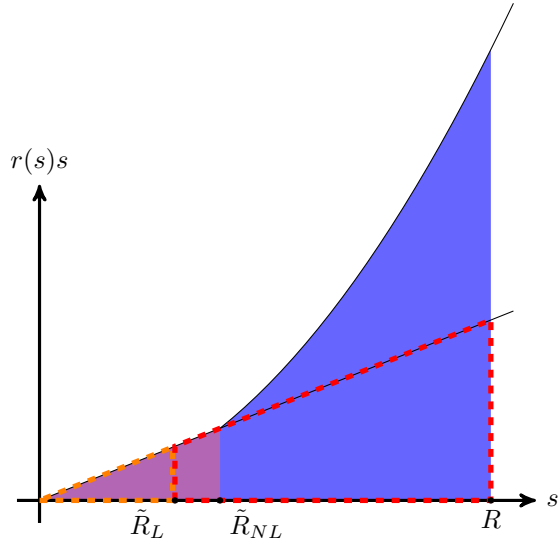
\begin{figure}[H]
\begin{center}
\begin{tikzpicture}[
axis/.style={very thick, ->, >=stealth'},scale=0.6]

  \fill[blue!60] (3,0) -- plot[domain=4:10] (\x,0.4/4*\x^2) --  (10,0);
  \fill[violet!60] (0,0) -- plot[domain=0:4](\x,0.4*\x) --  (4,0);
\draw[axis] (-0.5,0) -- (11,0) node [right] {$s$};
\draw[axis] (0,-0.5) -- (0,7) node [above] {$r(s)s$};
\draw[domain= 0:10.5, color=black] plot (\x,0.4*\x);
\draw[domain= 4:10.5, color=black] plot (\x,0.4/4*\x^2);
\draw[dashed, line width = 2pt, color=orange] (0,0) -- (2.95,0) -- (2.95,1.18) -- (0,0);
\draw[dashed, line width = 2pt, color=red] (3,0) -- (3,1.2) -- (10,4) -- (10,0) -- (3,0);
\fill[black] (3,0) circle (2pt) node[below left] {$\tilde{R}_{L}$};
\fill[black] (4,0) circle (2pt) node[below right] {$\tilde{R}_{NL}$};
\fill[black] (10,0) circle (2pt) node[below] {$R$};
\end{tikzpicture}
\caption{\small \sl $\tilde{R}_{L}$ and $\tilde{R}_{NL}$ and norms of the respective feedback control laws for $M=12$}
\end{center}
\end{figure}

\section{Invariance Analysis}\label{Invariance:analysis}

In what follows, we assume that the agents' initial states belong to a given domain $D\subset\Rat{n}$. In order to simplify the subsequent analysis, we assume that $D={\rm int}(B(\mathfrak{R}))$, namely the interior of the ball with center $0\in\Rat{n}$ and radius $\mathfrak{R}>0$. We aim at designing an appropriate modification of the feedback law \eqref{feedback:CaseI} which guarantees that the trajectories of the agents remain in $D$ for all future times.

\noindent For each $\varepsilon\in(0,\mathfrak{R})$, let $N_{\varepsilon}$ be the region with distance $\varepsilon$ from the boundary of $\partial D$ towards the interior of $D$, namely
\begin{equation}\label{subdomain:N}
N_{\varepsilon}:=\{x\in\Rat{n}:R-\varepsilon\le|x|<R\}
\end{equation}

\noindent and
\begin{equation}\label{subdomain:D}
D_{\varepsilon}:=D\setminus N_{\varepsilon}
\end{equation}

\noindent We proceed by defining a repulsive from the boundary of $D$ vector field, which when added to the dynamics of each agent in \eqref{feedback:CaseI}, will ensure the desired invariance of the closed loop system and simultaneously guarantee the same robust connectivity result established above. Let $h:[0,1]\to [0,1]$ be a Lipschitz continuous function that satisfies
\begin{equation} \label{function:h:properties}
h(0)=0;\;h(1)=1;\;h(\cdot)\;\textup{strictly increasing}
\end{equation}

\noindent We define the vector field $g:D\to\Rat{n}$ as
\begin{equation} \label{function:g}
g(x):=\left\lbrace\begin{array}{ll}-c\delta h\left(\frac{\varepsilon+|x|-R}{\varepsilon}\right)\frac{x}{|x|}, & {\rm if}\;x\in N_{\varepsilon} \\ 0, & {\rm if}\;x\in D_{\varepsilon} \end{array} \right.
\end{equation}

\noindent with $h(\cdot)$ as given above and appropriate positive constants $c$, $\delta$ which serve as design parameters. Then, it follows from \eqref{function:h:properties}, \eqref{function:g} and the Lipschitz property for $h(\cdot)$ that the vector field $g(\cdot)$ is Lipschitz continuous on $D$.

\noindent Having defined the mappings for the extra term in the dynamics of the modified controller which will guarantee the desired invariance property, we now state our main result.

\begin{prop}\label{Invariance:result}
For the multi-agent system \eqref{multiagent:single:integrator}, assume that  $D={\rm int}(B(\mathfrak{R}))$, for certain $\mathfrak{R}>0$ and that (ICH) is fulfilled. Furthermore, let $\varepsilon\in (0,\mathfrak{R})$, $N_{\varepsilon}$ and $D_{\varepsilon}$ as defined by \eqref{subdomain:N} and \eqref{subdomain:D}, respectively and assume that the initial states of all agents lie in $D_{\varepsilon}$. Then, there exists a control law \eqref{contol:law:connectedness} (with free inputs $v_{i}$) which guarantees both connectivity and invariance of $D$ for the solution of the system for all future times and is defined as
\begin{equation} \label{feedback:invariance}
u_{i}=g(x_{i})-\sum_{j\in\mathcal{N}_{i}}r(|x_{i}-x_{j}|)(x_{i}-x_{j})+v_{i}
\end{equation}

\noindent with $g(\cdot)$ given in \eqref{function:g} and certain $r(\cdot)$ satisfying Property (P). We choose the same positive constant $\delta$ in both \eqref{input:individual:bound} and \eqref{function:g} and select the constant $c$ in \eqref{function:g} greater that 1. Then the connectivity-invariance result is valid provided that the parameters $\delta$, $\tilde{R}$ and the function $r(\cdot)$ satisfy the restrictions \eqref{constant:delta:general:constraint}, \eqref{Pat:Rtilde:vs:R} and the input terms $v_i(\cdot)$, $i=1,\ldots,N$ satisfy \eqref{input:individual:bound}.
\end{prop}

\begin{proof}
We break the proof in two steps. In the first step, we show that as long as the invariance assumption is satisfied, namely, the solution of the closed loop system \eqref{multiagent:single:integrator}-\eqref{feedback:invariance} is defined and remains in $D$, network connectivity is maintained. In the second step, we show that for all times where the solution is defined, it remains inside a compact subset of $D$, which implies that the solution is defined and remains in $D$ for all future times, thus providing the desired invariance property.

\noindent \textbf{Step 1: Proof of network connectivity.}  \\

The proof of this step is based on an appropriate modification of the corresponding proof of Proposition \ref{connectivity:maintainance}. In particular, we exploit the energy function $V$ as given by \eqref{V:dfn} and show that when $|\Delta x|_{\infty}\ge\tilde{R}$, namely, when the maximum distance between two agents exceeds $\tilde{R}$ then its derivative along the solutions of the closed loop system is negative. Thus by using the same arguments with those in proof of Proposition  \ref{connectivity:maintainance} we can deduce that the system remains connected. Indeed, by evaluating the derivative of $V$ along the solutions of \eqref{multiagent:single:integrator}-\eqref{feedback:invariance} we obtain
\begin{align}
\dot{V}&=\sum_{i=1}^{N}\frac{\partial}{\partial x_{i}}V(x)\dot{x}_{i} \nonumber \\
&=\sum_{i=1}^{N}\frac{\partial}{\partial x_{i}}V(x)g(x_{i})-\sum_{l=1}^{n}c_{l}(x)^{T}L_{w}(x)^{2}c_{l}(x)+\sum_{l=1}^{n}c_{l}(x)^{T}L_{w}(x)c_{l}(v) \nonumber \\
&\le\sum_{i=1}^{N}\frac{\partial}{\partial x_{i}}V(x)g(x_{i})-\sum_{l=1}^{n}c_{l}(x)^{T}L_{w}(x)^{2}c_{l}(x)+\left|\sum_{l=1}^{n}c_{l}(x)^{T}L_{w}(x)c_{l}(v)\right| \label{Lyapunov:function:derivative:invariance}
\end{align}

\noindent By taking into account \eqref{Lyapunov:function:derivative} and using precisely the same arguments with those in proof of Steps 1 and 2 of Proposition \ref{connectivity:maintainance} it suffices to show that the first term of inequality \eqref{Lyapunov:function:derivative:invariance}, which by virtue of \eqref{V:partial:xi} is equal to
\begin{equation*}
\sum_{i=1}^{N}\frac{\partial}{\partial x_{i}}V(x)g(x_{i})=\sum_{i=1}^{N}\sum_{j\in\mathcal{N}_{i}}r(|x_{i}-x_{j}|)\langle (x_{i}-x_{j}),g(x_{i}) \rangle
\end{equation*}

\noindent is nonpositive for all $x\in D$. Given the partition $D_{\varepsilon}$, $N_{\varepsilon}$ of $D$, we consider for each agent $i\in\mathcal{N}$ the partition $\mathcal{N}_{i}^{D_{\varepsilon}}$, $\mathcal{N}_{i}^{N_{\varepsilon}}$ of its neighbors' set, corresponding to its neighbors that belong to $D_{\varepsilon}$ and $N_{\varepsilon}$, respectively. Also, we denote by $\mathcal{E}^{N_{\varepsilon}}$ the set of edges $\{i,j\}$ with both $x_{i},x_{j}\in N_{\varepsilon}$. Then, by taking into account that due to \eqref{function:g}, $g(x_{i})=0$ for $x_{i}\in D_{\varepsilon}$, it follows that
\begin{align}
&\sum_{i=1}^{N}\sum_{j\in\mathcal{N}_{i}}r(|x_{i}-x_{j}|)\langle (x_{i}-x_{j}),g(x_{i}) \rangle \nonumber \\
=&\sum_{\{i\in\mathcal{N}:x_{i}\in N_{\varepsilon}\}}\sum_{j\in\mathcal{N}_{i}}r(|x_{i}-x_{j}|)\langle (x_{i}-x_{j}),g(x_{i}) \rangle \nonumber \\
=&\sum_{\{i\in\mathcal{N}:x_{i}\in N_{\varepsilon}\}}\sum_{j\in\mathcal{N}_{i}^{D_{\varepsilon}}\cup\mathcal{N}_{i}^{N_{\varepsilon}}}r(|x_{i}-x_{j}|)\langle (x_{i}-x_{j}),g(x_{i}) \rangle \nonumber \\
=&\sum_{\{i\in\mathcal{N}:x_{i}\in N_{\varepsilon}\}}\sum_{j\in\mathcal{N}_{i}^{D_{\varepsilon}}}r(|x_{i}-x_{j}|)\langle (x_{i}-x_{j}),g(x_{i}) \rangle \nonumber \\
+&\sum_{\{i,j\}\in\mathcal{E}^{N_{\varepsilon}}}r(|x_{i}-x_{j}|)[\langle (x_{i}-x_{j}),g(x_{i}) \rangle+\langle (x_{j}-x_{i}),g(x_{j}) \rangle] \label{Lyapunov:function:derivative:extra:terms}
\end{align}

\noindent In order to prove that both terms in \eqref{Lyapunov:function:derivative:extra:terms} are less than or equal to zero and hence derive our desired result on the sign of $\dot{V}$, we exploit the following facts. \\

\noindent \textbf{Fact V.} Consider the vectors $\alpha,\beta,\gamma\in\Rat{n}$ with the following properties:
\begin{equation} \label{alpha:beta:normalized}
|\alpha|=1,\;|\beta|=1
\end{equation}
\begin{equation} \label{alpha:beta:amgle}
\langle \alpha,\gamma \rangle\ge 0,\;\langle \beta,\gamma \rangle\le 0
\end{equation}

\noindent Then for every quadruple $\lambda_{\alpha},\lambda_{\beta},\mu_{\alpha},\mu_{\beta}\in\RgeO$ satisfying
\begin{equation} \label{lambda:mu:inequalities}
\lambda_{\alpha}\ge\lambda_{\beta},\mu_{\alpha}\ge\mu_{\beta}
\end{equation}

\noindent it holds
\begin{equation} \label{alpha:beta:delta:angle}
\langle (\mu_{\alpha}\alpha-\mu_{\beta}\beta),\tilde{\delta} \rangle\ge 0
\end{equation}

\noindent where
\begin{equation} \label{vector:delta:dfn}
\tilde{\delta}:=\lambda_{\alpha}\alpha+\gamma-\lambda_{\beta}\beta
\end{equation}

\noindent We provide the proof of Fact V in the Appendix. \\

\noindent \textbf{Fact VI.} For any $x,\tilde{x}\in N_{\varepsilon}$ with $x=\lambda{\tilde{x}}$, $\lambda>0$ and $y\in {\rm cl}(D_{\varepsilon})$ it holds
\begin{equation*}
\langle (\tilde{x}-y),x \rangle\ge 0
\end{equation*}

\noindent The proof of Fact VI is based on the elementary properties $y\in {\rm cl}(D_{\varepsilon})\Rightarrow |y|\le \mathfrak{R}-\varepsilon$ and $x,\tilde{x}\in N_{\varepsilon}\Rightarrow \mathfrak{R}-\varepsilon\le |x|$ and $\mathfrak{R}-\varepsilon\le |\tilde{x}|$. Hence we have that
$$
\langle (\tilde{x}-y),x\rangle\ge |x||\tilde{x}|-|x||y|\ge 0
$$

\noindent We are now in position to show that both terms in the right hand side of \eqref{Lyapunov:function:derivative:extra:terms} are nonpositive, which according to our previous discussion establishes the desired connectivity maintenance result.

\noindent \textbf{Proof of the fact that the first term in \eqref{Lyapunov:function:derivative:extra:terms} is nonpositive.} For each $i,j$ in the first term in \eqref{Lyapunov:function:derivative:extra:terms} we get by applying Fact VI with $x,\tilde{x}=x_{i}\in N_{\varepsilon}$ and $y=x_{j}\in D_{\varepsilon}$ that
\begin{align*}
& r(|x_{i}-x_{j}|)\langle (x_{i}-x_{j}),g(x_{i}) \rangle  \\
= &  r(|x_{i}-x_{j}|)\frac{-c\delta h\left(\frac{\varepsilon+|x_{i}|-\mathfrak{R}}{\varepsilon}\right)}{|x_{i}|}\langle (x_{i}-x_{j}),x_{i} \rangle\le 0
\end{align*}

\noindent and hence that the first term is nonpositive.

\noindent \textbf{Proof of the fact that the second term in \eqref{Lyapunov:function:derivative:extra:terms} is nonpositive.}  We exploit Fact V in order to prove that for each $\{i,j\}\in\mathcal{E}^{N_{\varepsilon}}$ the quantity
\begin{equation}
\langle (x_{i}-x_{j}),g(x_{i}) \rangle+\langle (x_{j}-x_{i}),g(x_{j}) \rangle
\end{equation}

\noindent in the second term of \eqref{Lyapunov:function:derivative:extra:terms} is nonpositive as well. Notice that both $x_i,x_j\in N_{\varepsilon}$ and without loss of generality we may assume that
\begin{equation} \label{distance:xixj:boundary}
|x_{i}|\ge |x_{j}|
\end{equation}

\noindent namely, that $x_{i}$ is farther from the boundary of $D_{\varepsilon}$ than $x_{j}$. Then by setting
\begin{align}
\alpha:= & \frac{x_{i}}{|x_{i}|} \label{alpha:dfn} \\
\beta:= & \frac{x_{j}}{|x_{j}|} \label{beta:dfn} \\
\gamma:= & \tilde{x}_{i}-\tilde{x}_{j} \label{gamma:dfn}
\end{align}

\noindent with
\begin{align}
\tilde{x}_{i}:= & x_{i}-(|x_i|+\varepsilon-\mathfrak{R})\frac{x_{i}}{|x_{i}|} \label{tilde:xi:dfn} \\
\tilde{x}_{j}:= & x_{j}-(|x_j|+\varepsilon-\mathfrak{R})\frac{x_{j}}{|x_{j}|} \label{tilde:xj:dfn}
\end{align}

\noindent and
\begin{align}
\lambda_{\alpha}:= & |x_i|+\varepsilon-\mathfrak{R} \label{lambda:alpha:dfn} \\
\lambda_{\beta}:= & |x_j|+\varepsilon-\mathfrak{R} \label{lambda:beta:dfn} \\
\mu_{\alpha}:= & c\delta h\left(|x_i|+\varepsilon-\mathfrak{R}\right) \label{mu:alpha:dfn} \\
\mu_{\beta}:= & c\delta h\left(|x_j|+\varepsilon-\mathfrak{R}\right) \label{mu:beta:dfn}
\end{align}

\noindent it follows from \eqref{alpha:dfn} and \eqref{beta:dfn} that
\begin{equation*}
|\alpha|=|\beta|=1
\end{equation*}

\noindent and from \eqref{function:h:properties}, \eqref{distance:xixj:boundary}, \eqref{lambda:alpha:dfn}, \eqref{lambda:beta:dfn}, \eqref{mu:alpha:dfn} and \eqref{mu:beta:dfn} that
\begin{equation*}
\lambda_{\alpha}\ge\lambda_{\beta}\ge 0,\mu_{\alpha}\ge\mu_{\beta}\ge 0
\end{equation*}

\noindent Furthermore, we get from \eqref{tilde:xi:dfn} and \eqref{tilde:xj:dfn} that $|\tilde{x}_i|=|\tilde{x}_j|=\mathfrak{R}-\varepsilon\Rightarrow\tilde{x}_{i},\tilde{x}_{j}\in\partial D_{\varepsilon}$. Thus, it follows from \eqref{alpha:dfn}, \eqref{beta:dfn}, \eqref{gamma:dfn} and application of Fact VI with $x=x_i$, $\tilde{x}=\tilde{x}_i$ and $y=\tilde{x}_j$ that
\begin{equation*}
\langle \alpha,\gamma \rangle\ge 0
\end{equation*}

\noindent and similarly that
\begin{equation*}
\langle \beta,\gamma \rangle\le 0
\end{equation*}

\noindent It follows that all requirements of Fact V are fulfilled. Furthermore, by taking into account \eqref{alpha:dfn}-\eqref{lambda:beta:dfn}, we get that
\begin{equation}\label{tilde:delta}
\tilde{\delta}=\lambda_{\alpha}\alpha+\gamma-\lambda_{\beta}\beta=x_{i}-x_{j}
\end{equation}

\noindent Thus we establish by virtue of \eqref{function:g}, \eqref{alpha:beta:delta:angle}, \eqref{vector:delta:dfn}, \eqref{alpha:dfn}, \eqref{beta:dfn}, \eqref{mu:alpha:dfn}, \eqref{mu:beta:dfn} and \eqref{tilde:delta}  that
\begin{align*}
\langle (\mu_{\alpha}\alpha-\mu_{\beta}\beta),\tilde{\delta} \rangle=-\langle (g(x_{i})-g(x_{j})),(x_{i}-x_{j}) \rangle\ge 0 & \iff \\
\langle (x_{i}-x_{j}),g(x_{i}) \rangle+\langle (x_{j}-x_{i}),g(x_{j}) \rangle\le 0 &
\end{align*}

\noindent as desired.

\noindent \textbf{Step 2: Proof of forward invariance of $D$ with respect to the solution of \eqref{multiagent:single:integrator}-\eqref{feedback:invariance}.}  \\

We proceed by proving that the control law \eqref{feedback:invariance} also guarantees the desired invariance property for the solutions of system \eqref{multiagent:single:integrator}-\eqref{feedback:invariance}, provided that the input terms $v_i(\cdot)$, $i=1,\ldots,N$ satisfy \eqref{input:individual:bound}. Let $[0,T_{\max})$ be the maximal forward interval for which the solution $x(\cdot)$ of  \eqref{multiagent:single:integrator}-\eqref{feedback:invariance} with $x(0)\in D_{\varepsilon}$ exists and remains inside $D$. We claim that for all $t\in [0,T_{\max})$ the solution remains inside ${\rm cl}(D_{\frac{\tilde{c}-1}{\tilde{c}}\varepsilon})$ with
\begin{equation}\label{constant:tilde:c}
\tilde{c}=\frac{1}{h^{-1}\left(\frac{1}{c}\right)}\iff h\left(\frac{1}{\tilde{c}}\right)=\frac{1}{c};\tilde{c}>1
\end{equation}

\noindent and where $c>1$ and $h(\cdot)$ are given in the statement of the proposition and \eqref{function:h:properties}, respectively. Then, it also follows from the fact that $x(t)$ remains in the compact subset ${\rm cl}(D_{\frac{\tilde{c}-1}{\tilde{c}}\varepsilon})$ of $D$ for all $t\in [0,T_{\max})$, that $T_{\max}=\infty$, which provides the desired result. In order to prove our claim, we need to define certain auxiliary mappings. For each $i\in\{1,\ldots,N\}$ we define the functions
\begin{equation} \label{functions:mi}
m_{i}(t):=\left\lbrace \begin{array}{ll} \varepsilon+|x_{i}(t)|-\mathfrak{R}, & {\rm if}\; x_{i}(t)\in N_{\varepsilon} \\ 0, & {\rm if}\; x_{i}(t)\in D_{\varepsilon} \end{array} \right., \quad t\in [0,T_{\max})
\end{equation}

\noindent and
\begin{equation} \label{function:m}
m(t):=\max\{m_{i}(t):i=1,\ldots,N\},t\in [0,T_{\max})
\end{equation}

\noindent where $m_{i}(t)$ denotes the distance of agent $i$ from $D_{\varepsilon}$ at time $t$ and $m(t)$ is the maximum over those distances for all agents. Hence, for all $t\in [0,T_{\max})$ and all $ \bar{\varepsilon}\in(0,\varepsilon]$ we have the following equivalences
\begin{align}
x_{i}(t)\in N_{\bar{\varepsilon}}&\iff m_{i}(t)\in [\varepsilon-\bar{\varepsilon},\varepsilon) \label{xi:in:Nepsilon} \\
x_{i}(t)\in \partial D_{\bar{\varepsilon}}&\iff m_i(t)=\varepsilon-\bar{\varepsilon} \label{xi:in:bndryDepsilon}
\end{align}

\noindent and for all $i=1,\ldots,N $ that
\begin{equation}\label{xi:in:clsDepsilon}
x_i(t)\in{\rm cl}( D_{\bar{\varepsilon}}),\forall i=1,\ldots,N \iff m(t)\in[0,\varepsilon-\bar{\varepsilon}]
\end{equation}

\noindent Notice that the functions $m_{i}(\cdot)$, $i=1,\ldots,N$ and $m(\cdot)$ are continuous and due to our hypothesis that $x(0)\in D_{\varepsilon}$, satisfy
\begin{equation}\label{function:m:at0}
m(0)=0
\end{equation}

\noindent We claim that
\begin{equation}\label{invariance:equivalent}
m(t)\le \frac{\varepsilon}{\tilde{c}}, \forall t\in [0,T_{\max})
\end{equation}

\noindent with $\tilde{c}$ as given in \eqref{constant:tilde:c}. Indeed, suppose on the contrary that there exists $T\in (0,T_{\max})$ such that
\begin{equation}\label{function:m:atT}
m(T)=\tilde{\varepsilon}\in\left(\frac{\varepsilon}{\tilde{c}},\varepsilon\right)
\end{equation}

\noindent and define
\begin{equation} \label{time:tau:invariance}
\tau:=\min\left\lbrace\tilde{\tau}\in[0,T]:m(t)\ge\frac{1}{2}\left(\tilde{\varepsilon}+\frac{\varepsilon}{\tilde{c}}\right),\forall t\in[\tilde{\tau},T]\right\rbrace
\end{equation}

\noindent Then it follows from \eqref{function:m:atT} that $\tau$ is well defined and from \eqref{function:m:at0}, \eqref{time:tau:invariance} and the continuity of $m(\cdot)$ that
\begin{equation}\label{function:m:attau}
m(\tau)=\frac{1}{2}\left(\tilde{\varepsilon}+\frac{\varepsilon}{\tilde{c}}\right)
\end{equation}

\noindent and that there exists a sequence $(t_{\nu})_{\nu\in\mathbb{N}}$ with
\begin{equation} \label{sequence:tnu:properties}
t_{\nu}\searrow\tau\;{\rm and}\;m(t_{\nu})\ge\frac{1}{2}\left(\tilde{\varepsilon}+\frac{\varepsilon}{\tilde{c}}\right),\forall \nu\in\mathbb{N}
\end{equation}

\noindent From \eqref{function:m}, \eqref{function:m:attau}, \eqref{sequence:tnu:properties} and the infinite pigeonhole principle we deduce that there exists $i\in\{1,\ldots,N\}$ and a subsequence $(t_{\nu_k})_{k\in\mathbb{N}}$ of $(t_{\nu})_{\nu\in\mathbb{N}}$ such that
\begin{equation}  \label{subsequence:tnu:mis}
m_{i}(t_{\nu_k})\ge\frac{1}{2}\left(\tilde{\varepsilon}+\frac{\varepsilon}{\tilde{c}}\right),\forall k\in\mathbb{N}; m_i(\tau)=\frac{1}{2}\left(\tilde{\varepsilon}+\frac{\varepsilon}{\tilde{c}}\right)
\end{equation}

\noindent Thus, it follows by virtue of \eqref{xi:in:Nepsilon} and \eqref{xi:in:bndryDepsilon} that

\begin{equation}  \label{subsequence:tnu:properties}
x_{i}(t_{\nu_k})\in N_{\varepsilon-\frac{1}{2}\left(\tilde{\varepsilon}+\frac{\varepsilon}{\tilde{c}}\right)},\forall k\in\mathbb{N};x_i(\tau)\in\partial D_{\varepsilon-\frac{1}{2}\left(\tilde{\varepsilon}+\frac{\varepsilon}{\tilde{c}}\right)}
\end{equation}

\noindent Set $W(x)=|x|^2$ and notice that due to \eqref{subsequence:tnu:properties} it holds $|x_i(t_{\nu_k})|\ge|x_i(\tau)|$, $\forall k\in \mathbb{N}$. The latter implies that
\begin{equation} \label{Wtilde:derivative:case1}
\left.\frac{d}{dt}W(x_{i}(t))\right|_{t=\tau}=\lim_{k\to\infty}\frac{W(x_{i}(t_{\nu_k}))-W(x_{i}(\tau))}{t_{\nu_k}-\tau}\ge 0
\end{equation}

\noindent On the other hand, we have that
\begin{align} \label{Wtilde:derivative:case2}
\left.\frac{d}{dt}W(x_{i}(t))\right|_{t=\tau}&=\nabla W(x_{i}(\tau))\dot{x}_{i}(\tau) \nonumber \\
&=x_{i}(\tau)^T\left[g(x_{i}(\tau))+v_{i}(\tau)-\sum_{j\in\mathcal{N}_{i}}r(|x_{i}(\tau)-x_{j}(\tau)|)(x_{i}(\tau)-x_{j}(\tau))\right]
\end{align}

\noindent By taking into account \eqref{function:m:atT} and \eqref{subsequence:tnu:properties} we get that
$$
|x_i(\tau)|=\mathfrak{R}-\varepsilon+\frac{1}{2}\left(\tilde{\varepsilon}+\frac{\varepsilon}{\tilde{c}}\right)>\mathfrak{R}-\varepsilon+\frac{\varepsilon}{\tilde{c}}=\mathfrak{R}-\frac{\tilde{c}-1}{\tilde{c}}\varepsilon
$$

\noindent and hence from \eqref{function:h:properties}, \eqref{function:g} and \eqref{constant:tilde:c} that
\begin{equation} \label{eq1}
|g(x_i(\tau))|>c\delta h\left(\frac{\varepsilon+\mathfrak{R}-\frac{\tilde{c}-1}{\tilde{c}}\varepsilon-\mathfrak{R}}\varepsilon\right)=c\delta h\left(1-\frac{\tilde{c}-1}{\tilde{c}}\right)=c\delta h\left(\frac{1}{\tilde{c}}\right)=\delta
\end{equation}

\noindent Also, due to \eqref{function:g} it holds
\begin{equation}\label{eq2}
x_{i}(\tau)=-ag(x_{i}(\tau))
\end{equation}

\noindent for certain $a>0$. Then we get from \eqref{eq1}, \eqref{eq2} and the fact that $|v_{i}(\tau)|\le\delta$ that
\begin{align} \label{Wtilde:derivative:case2:term1}
x_{i}(\tau)^T[g(x_{i}(\tau))+v_{i}(\tau)] & \le x_{i}(\tau)^T g(x_{i}(\tau))+|x_{i}(\tau)||v_{i}(\tau)| \nonumber \\
& =-|x_{i}(\tau)||g(x_{i}(\tau))|+|x_{i}(\tau)||v_{i}(\tau)| \nonumber \\
& =-|x_{i}(\tau)|(|g(x_{i}(\tau))|-|v_{i}(\tau)|)<0
\end{align}

\noindent Furthermore, we have from \eqref{function:m:attau} and \eqref{xi:in:clsDepsilon} that $x_{j}(\tau)\in {\rm cl}(D_{\varepsilon-\frac{1}{2}\left(\tilde{\varepsilon}+\frac{\varepsilon}{\tilde{c}}\right)})$ for all $j\in\mathcal{N}_i$ and from \eqref{subsequence:tnu:properties} that  $x_{i}(\tau)\in N_{\varepsilon-\frac{1}{2}\left(\tilde{\varepsilon}+\frac{\varepsilon}{\tilde{c}}\right)}$. Thus, it follows from Fact VI that
\begin{equation} \label{Wtilde:derivative:case2:term2}
x_{i}(\tau)^T(x_{i}(\tau)-x_{j}(\tau))=\langle x_{i}(\tau),(x_{i}(\tau)-x_{j}(\tau)) \rangle \le 0
\end{equation}

\noindent From \eqref{Wtilde:derivative:case2:term1} and \eqref{Wtilde:derivative:case2:term2}, we obtain that \eqref{Wtilde:derivative:case2} is negative, which contradicts \eqref{Wtilde:derivative:case1}. Hence, \eqref{invariance:equivalent} holds, which implies that $x(t)$ remains in the compact subset ${\rm cl}(D_{\frac{\tilde{c}-1}{\tilde{c}}\varepsilon})$ of $D$ for all $t\in [0,T_{\max})$. Thus, $T_{\max}=\infty$ and we conclude that the solution $x(\cdot)$ of the system remains in $D$ for all $t\ge 0$.
\end{proof}

\section{Conclusions}

We have provided a distributed control scheme which guarantees connectivity of a multi-agent network governed by single integrator dynamics. The corresponding control law is robust with respect to additional free input terms which can further be exploited for motion planning. For the case of a spherical domain, adding a repulsive vector field near the boundary ensures that the agents remain inside the domain for all future times. The latter framework is motivated by the fact that it allows us to abstract the behaviour of the system through a finite transition system and exploit formal method tools for high level planning.

Further research directions include the generalization of the invariance result of Section \ref{Invariance:analysis} for the case where the domain is convex and has smooth boundary and the improvement of the bound on the free input terms, by allowing the bound to be state dependent.

\section{Appendix}

In the Appendix, we provide the proofs of Facts I, II, III and IV which were used in proof of Proposition \ref{connectivity:maintainance} and of Fact V, in proof of Proposition \ref{Invariance:result}. For convenience we state the elementary inequality
\begin{equation} \label{parallelogram:law:consequence}
2(|w|^{2}+|z|^{2})\ge |w-z|^{2},\forall w,z\in\Rat{n}
\end{equation}

\noindent \textbf{Proof of Fact I.} Let $\{e_k\}_{k=1,\ldots,N}$ be an orthonormal basis of eigenvectors corresponding to the ordered eigenvalues of $L_w(x)$. Then, for each $l=1,\ldots,n$ we have that
$$
c_l(x^{\perp})=\sum_{k=2}^N\mu_k e_k;\mu_k\in\Rat{},k=2,\ldots,N
$$

\noindent and hence, that
$$
|c_l(x^{\perp})|=\left(\sum_{k=2}^N\mu_k^2\right)^{\frac{1}{2}}
$$

\noindent Thus, we get that
\begin{align*}
|L_w(x)c_l(x^{\perp})|&=\left|\sum_{k=2}^N\mu_k L_w(x)e_k\right|=\left|\sum_{k=2}^N\mu_k \lambda_k(x)e_k\right| \\
&=\left(\sum_{k=2}^N(\mu_k \lambda_k(x))^2\right)^{\frac{1}{2}}\ge\lambda_2(x)\left(\sum_{k=2}^N\mu_k^2\right)^{\frac{1}{2}}=\lambda_2(x)|c_l(x^{\perp})|
\end{align*}

\noindent which establishes \eqref{fact1}.

\noindent \textbf{Proof of Fact II.} By taking into account the Cauchy Schwartz inequality we obtain
\begin{align*}
\sum_{l=1}^n|c_l(x)||c_l(y)|&\le\left(\sum_{l=1}^n|c_l(x)|^2\right)^{\frac{1}{2}}\left(\sum_{l=1}^n|c_l(y)|^2\right)^{\frac{1}{2}} \\
&=\left(\sum_{l=1}^n\sum_{i=1}^N c_l(x_i)^2\right)^{\frac{1}{2}}\left(\sum_{l=1}^n\sum_{i=1}^N c_l(y_i)^2\right)^{\frac{1}{2}} \\
&=\left(\sum_{i=1}^N\sum_{l=1}^n c_l(x_i)^2\right)^{\frac{1}{2}}\left(\sum_{i=1}^N\sum_{l=1}^n c_l(y_i)^2\right)^{\frac{1}{2}} \\
&=\left(\sum_{i=1}^N|x_i|^2\right)^{\frac{1}{2}}\left(\sum_{l=1}^N|y_i|^2\right)^{\frac{1}{2}}=|x||y|
\end{align*}
\noindent and hence \eqref{fact2} holds.

\noindent \textbf{Proof of Fact III.} By the definition of $x^{\perp}$ and $\bar{x}$, it follows that there exists $\tilde{x}\in\Rat{n}$ such that $x-x^{\perp}=\bar{x}=(\tilde{x},\ldots,\tilde{x})\in\Rat{Nn}$. Hence, we have that
\begin{align*}
|x^{\perp}|=|x-\bar{x}|&=|(x_{1},\ldots,x_{N})-(\tilde{x},\ldots,\tilde{x})|=|(x_{1}-\tilde{x},\ldots,x_{N}-\tilde{x})| \\
&=\left(\sum_{i=1}^{N}|x_{i}-\tilde{x}|^{2}\right)^{\frac{1}{2}} \Rightarrow \\
\sqrt{2(N-1)}|x^{\perp}|&=\left(\sum_{i=1}^{N}2(N-1)|x_{i}-\tilde{x}|^{2}\right)^{\frac{1}{2}}=\left(\sum_{\{i,j\}\in \mathcal{E}(K(\mathcal{N}))}2(|x_{i}-\tilde{x}|^{2}+|x_{j}-\tilde{x}|^{2})\right)^{\frac{1}{2}}
\end{align*}

\noindent where $\mathcal{E}(K(\mathcal{N}))$ stands for the edge set of the complete graph with vertex set $\mathcal{N}$. Then, it follows from \eqref{parallelogram:law:consequence} that
\begin{align*}
\left(\sum_{\{i,j\}\in\mathcal{E}(K(\mathcal{N}))}2(|x_{i}-\tilde{x}|^{2}+|x_{j}-\tilde{x}|^{2})\right)^{\frac{1}{2}}&\ge \left(\sum_{\{i,j\}\in\mathcal{E}(K(\mathcal{N}))}|x_{i}-x_{j}|^{2}\right)^{\frac{1}{2}} \\
&\ge \left(\sum_{\{i,j\}\in\mathcal{E}}|x_{i}-x_{j}|^{2}\right)^{\frac{1}{2}}=|\Delta x|
\end{align*}

\noindent which provides the desired result.

\noindent \textbf{Proof of Fact IV.} Notice that \eqref{fact4} is equivalently written as
\begin{align*}
\sqrt{2}|x^{\perp}|\ge \max_{\{i,j\}\in\mathcal{E}}|x_{i}-x_{j}|&\iff 2|x^{\perp}|\ge \max_{\{i,j\}\in\mathcal{E}}|x_{i}-x_{j}|^{2} \\
&\iff 2\left(\sum_{i=1}^{N}|x_{i}-\tilde{x}|^{2}\right)\ge \max_{\{i,j\}\in\mathcal{E}}|x_{i}-x_{j}|^{2}
\end{align*}

\noindent with $\tilde{x}\in\Rat{n}$  as in proof of Fact III. Let $\{\hat{i},\hat{j}\}\in\mathcal{E}$ such that $|x_{\hat{i}}-x_{\hat{j}}|=\max_{\{i,j\}\in\mathcal{E}}|x_{i}-x_{j}|$. Then, by taking into account \eqref{parallelogram:law:consequence} we have
\begin{align*}
2\left(\sum_{i=1}^{N}|x_{i}-\tilde{x}|^{2}\right)\ge 2(|x_{\hat{i}}-\tilde{x}|^{2}+|x_{\hat{j}}-\tilde{x}|^{2})\ge |x_{\hat{i}}-x_{\hat{j}}|^{2}=\max_{\{i,j\}\in\mathcal{E}}|x_{i}-x_{j}|^{2}
\end{align*}

\noindent and thus \eqref{fact4} is fulfilled.

\noindent \textbf{Proof of Fact V.} By taking into account \eqref{alpha:beta:normalized}-\eqref{lambda:mu:inequalities} and \eqref{vector:delta:dfn} we evaluate
\begin{align*}
\langle (\mu_{\alpha}\alpha-\mu_{\beta}\beta),\tilde{\delta} \rangle & =\langle (\mu_{\alpha}\alpha-\mu_{\beta}\beta),(\lambda_{\alpha}\alpha-\lambda_{\beta}\beta)+\gamma\rangle \\
& = \langle (\mu_{\alpha}\alpha-\mu_{\beta}\beta),(\lambda_{\alpha}\alpha-\lambda_{\beta}\beta)\rangle+\mu_{\alpha}\langle \alpha,\gamma \rangle-\mu_{\beta}\langle \beta,\gamma \rangle \\
& \ge \langle (\mu_{\alpha}\alpha-\mu_{\beta}\beta),(\lambda_{\alpha}\alpha-\lambda_{\beta}\beta)\rangle \\
& = \mu_{\alpha}\lambda_{\alpha}|\alpha|^{2}-(\mu_{\alpha}\lambda_{\beta}+\mu_{\beta}\lambda_{\alpha})\langle \alpha,\beta \rangle+\mu_{\beta}\lambda_{\beta}|\beta|^{2} \\
& \ge \mu_{\alpha}\lambda_{\alpha}|\alpha|^{2}-(\mu_{\alpha}\lambda_{\beta}+\mu_{\beta}\lambda_{\alpha})|\alpha||\beta|+\mu_{\beta}\lambda_{\beta}|\beta|^{2} \\
& = \mu_{\alpha}\lambda_{\alpha}-\mu_{\alpha}\lambda_{\beta}-\mu_{\beta}\lambda_{\alpha}+\mu_{\beta}\lambda_{\beta} \\
& = \mu_{\alpha}(\lambda_{\alpha}-\lambda_{\beta})-\mu_{\beta}(\lambda_{\alpha}-\lambda_{\beta})=(\mu_{\alpha}-\mu_{\beta})(\lambda_{\alpha}-\lambda_{\beta})\ge 0 \\
\end{align*}

\noindent and hence \eqref{alpha:beta:delta:angle} holds.

\section{Acknowledgements}

This work was supported by the EU STREP RECONFIG: FP7-ICT-2011-9-600825, the H2020 ERC Starting Grant BUCOPHSYS and the Swedish Research Council (VR).

\end{document}